\documentclass{ifacconf}

\usepackage{graphicx}      % include this line if your document contains figures
\usepackage{natbib}        % required for bibliography
\usepackage{amsmath, amssymb, amsfonts} % Nicer mathematical typesetting
\usepackage{xcolor} % Predefines additional colors and allows user defined colors
\usepackage{comment} %add block comments

\definecolor{matlabblue}{rgb}{0, 0.4470, 0.7410}
\definecolor{matlabred}{rgb}{0.8500, 0.3250, 0.0980}
\definecolor{matlabgray}{rgb}{0.5, 0.5, 0.5}
\definecolor{matlablg}{rgb}{0.75, 0.75, 0.75}
\definecolor{matlabyellow}{rgb}{0.9290, 0.6940, 0.1250}
\definecolor{matlabpurple}{rgb}{0.4940, 0.1840, 0.5560}
\definecolor{matlab_br_red}{rgb}{1, 0, 0}
\definecolor{matlab_br_blue}{rgb}{0, 0, 1}
\definecolor{matlabblack}{rgb}{0, 0, 0}

\begin{document}
\begin{frontmatter}

\title{Efficient Learning of Affine and Rational Dependency LPV Models With Linear Fractional Representation\thanksref{footnoteinfo}} 

\thanks[footnoteinfo]{This work is funded by the European Union (ERC, COMPLETE, 101075836), by the Air Force Office of Scientific Research under award number FA8655-23-1-7061, and the Mathworks Inc. Views and opinions expressed are however those of the author(s) only and do not necessarily reflect those of the European Union or the European Research Council Executive Agency. Neither the European Union nor the granting authority can be held responsible for them.}

\author[First]{Roel Drenth} 
\author[First]{Jan H. Hoekstra} 
\author[First]{Maarten Schoukens} 
\author[First,Second]{Roland T\'{o}th}

\address[First]{Control Systems Group, Dept. of Electrical Eng., Eindhoven University of Technology, Eindhoven, The Netherlands.(e-mail: \{r.drenth, j.h.hoekstra, m.schoukens, r.toth\}@tue.nl).}
\address[Second]{Systems and Control Laboratory, HUN-REN Institute for Computer Science and Control, Budapest, Hungary}
\begin{abstract} 
Identifying control-friendly models of nonlinear systems remains one of the major challenges at the intersection of system identification and control. The \emph{Linear Parameter-Varying} (LPV) framework offers a promising solution, but existing identification methods often rely on model structures with affine scheduling dependency. Instead, this work proposes the use of LPV models with \emph{Linear Fractional Representation} (LFR) admitting a rational scheduling-dependency, capable of modelling complex nonlinear systems with fewer scheduling variables compared to affine models. This work introduces a direct parameterization to ensure well-posedness of rational LPV-LFR models, which by joint-estimation of an LPV plant and scheduling map, using only input-output data, is capable of modelling complex nonlinear systems. Accuracy of the proposed approach is shown on two simulation examples.
\end{abstract}

\end{frontmatter}
%===============================================================================

\section{Introduction}
As the performance demands of systems in engineering grow, the limitations of the ubiquitous \emph{Linear Time-Invariant} (LTI) assumptions become apparent. Employing the LTI framework to model systems which are often inherently \emph{nonlinear} (NL) results in models with limited accuracy, which subsequently limits achievable performance in (model-based) control.

The \emph{Linear Parameter-Varying} (LPV) framework offers an attractive alternative to the LTI framework. LPV models retain a structure that is linear in the input, but the model parameters can vary as a function of a \emph{scheduling variable}. Through change in the scheduling variable, LPV models are capable of representing time-varying and nonlinear behaviour. This results in a model class which provides a middle-ground between the limitations of the LTI class and the complexity of nonlinear models. Importantly, the LPV structure enables controller synthesis tools with performance guarantees (see e.g., \cite{mohammadpour2012}). Moreover, the LPV framework enjoys the support of various software toolboxes, implementing powerful control and modelling approaches \citep{hjartarson2015, boef2021}.

Traditional LPV identification methods, which require knowledge of the scheduling signal, are often extensions of their LTI counterparts. Notable methods include \emph{Prediction Error Minimization} (PEM)  methods (see e.g., \cite{toth2010}), and \emph{Subspace Identification} (see e.g., \cite{vanwingerden2009}, \cite{cox2021}). In recent years, methods which jointly identify an LPV model and its scheduling map have been gaining attention. By estimating the scheduling map directly from data, the often significant challenge of modelling the scheduling is bypassed.  In particular, \emph{deep-learning}-based identification methods with LPV-\emph{state-space} (SS) models and \emph{Neural Network} (NN) scheduling maps have been shown to efficiently estimate LPV/NL models \citep{verhoek2022} and noise processes \citep{bemporad2025b}.

LPV identification methods have been developed for a wide range of model classes, including LPV-\emph{Input-Output} (IO), LPV-SS and LPV-\emph{Linear Fractional Representation} (LFR) classes. However, the LPV-SS and LPV-LFR model classes are most commonly used for controller synthesis. For LPV-SS models, controller synthesis methods require affine dependency on the scheduling variables. Using a simple affine-dependency structure to model complex nonlinearities often comes with the trade-off of requiring increased number of scheduling variables and a complex scheduling map, which in turn can introduce unnecessary conservatism in controller synthesis, called overbounding. Alternatively, we consider the LPV-LFR class of models. This class admits rational dependency on the scheduling variable which can reduce the required dimension and complexity of the scheduling map and thus the effects of overbounding (see e.g., \cite{hoffmann2014} for a comparison). Importantly, this model class is still suited for controller synthesis  (\cite{scherer2001}). However, identification of LPV-LFR models corresponding to rational-dependency is a significantly more challenging problem as the rational terms can introduce singularities.

So far, little research on the identification of LPV-LFR models has been performed, and the methods that do exist face significant challenges, such as requiring full-state measurements (\cite{nemani1995}), being restricted to affine-dependency (\cite{mejari2019,bianchi2010}), or assuming models cannot be ill-posed (\cite{lee1999}, \cite{cheng2015}). In particular, to the best of the authors' knowledge, no methods to ensure well-posedness in identification of LPV-LFR models under rational scheduling dependency exist. Furthermore, no research on the use of LPV-LFR models in the aforementioned joint scheduling and model identification setting has been performed.

In this work we leverage the efficiency of \emph{deep-learning}-based joint-estimation methods with the expressivity of rational LPV-LFR models to estimate nonlinear models. Specifically, as main contribution, we provide a direct parameterization of LPV-LFR models that guarantees well-posedness, inspired by direct parameterizations for stability of recurrent models (\cite{revay2020}, \cite{verhoek2023}).

The remainder of this paper first discusses the problem setting, including the model structure, in Section \ref{sec: ident_setting}. Next, in Section \ref{sec: WP}, well-posedness conditions and a direct-parameterization satisfying such conditions is introduced. This is followed by the introduction of the identification methodology in Section \ref{sec: ident_method}. Finally, the proposed method is evaluated on simulation examples in Section \ref{sec: sim_examples}, after which conclusions are drawn in Section \ref{sec: conclusion}.
\section{Problem Setting} \label{sec: ident_setting}
\subsection{Data-generating system}
We consider \emph{discrete-time} (DT) data-generating systems with the following state-space representation:
\begin{equation}
\label{NLsys_eq}
\Sigma_\mathrm{o}: \left\{
\begin{aligned}
x_{k+1}  &= f(x_k,u_k),\\
y_k &= g(x_k,u_k) + e_k,
\end{aligned}
\right.
\end{equation}

where $k \in \mathbb{Z}$ is the discrete time, where $x_k \in  \mathbb{X}\subseteq \mathbb{R}^{n_\mathrm{x}}$ is the state of the plant,  $u_k \in \mathbb{U} \subseteq \mathbb{R}^{n_\mathrm{u}}$ is the input signal, $y_k \in \mathbb{Y} \subseteq \mathbb{R}^{n_\mathrm{y}}$ is the system output, with output-additive white noise noise  $e_k \in \mathbb{E} \subseteq \mathbb{R}^{n_\mathrm{y}}$.  $u_k$ is assumed to be quasi-stationary. The state transition function $f: \mathbb{X} \times \mathbb{U} \!\to\! \mathbb{X}$ and output function $g: \mathbb{X} \times \mathbb{U} \!\to\! \mathbb{Y}$ are considered to be deterministic and nonlinear.

\subsection{Linear Fractional Representation}
The \emph{Linear Fractional Representation} is a general model representation that allows one to model dynamical systems with a wide range of properties, such as nonlinearities and uncertainty, in a unified manner. It originally gained popularity in the field of robust control (\cite{zhou1996}) where it is used to model uncertain systems. For LPV systems it is instead used to incorporate the parameter-varying nature. With an LFR, LPV systems are represented by interconnections between nominal systems $M$ and their parameter-varying components represented by matrix functions $\Delta(p)$ with linear dependence on the scheduling signal $p$ \citep{schererLMIbook}.

Let us define an LPV-LFR as a pair $\{M, \Delta(p)\}$ where the interconnection in discrete time follows the standard form:
\begin{equation}
\label{LPVLFR_eq}
\begin{aligned}
M: \left\{
\begin{aligned}
\begin{bmatrix}
x_{k+1} \\
z_k\\
y_k
\end{bmatrix} &= \begin{bmatrix}
A_x & B_w & B_u\\
C_z & D_{zw} & D_{zu}\\
C_y & D_{yw} & D_{yu}
\end{bmatrix}
\begin{bmatrix}
x_k\\
w_k\\
u_k
\end{bmatrix}\\
\end{aligned}
\right.,\\
w_k = \Delta(p_k)z_k, \qquad \qquad \quad \ 
\end{aligned}
\end{equation}  
where $w_k, z_k \in \mathbb{R}^{n_\mathrm{w}}$ are the latent variables, $A_x, \dots D_{yu}$ are real matrices of appropriate dimensions, and $p_k \in \mathbb{P} \subseteq\mathbb{R}^{n_\mathrm{p}}$ is the scheduling variable. The matrix function $\Delta(p_k): \mathbb{P} \to \mathbb{R}^{n_\mathrm{w} \times n_\mathrm{w}}$, commonly referred to as the Delta-block, is typically assumed to be diagonal and can be defined as:
\begin{equation}
\label{Delta_structure}
    \Delta(p_k) = \begin{bmatrix}
        p_{1,k}I_{\eta_1} &   & 0\\
         & \ddots & \\
        0 &  & p_{n_\mathrm{p},k}I_{\eta_{n_\mathrm{p}}}
    \end{bmatrix},
\end{equation}
where $\eta$ is a vector of integers denoting the number of repetitions per scheduling variable and where $I_{\eta_1},\dots,I_{\eta_{n_\mathrm{p}}}$ are identity matrices of appropriate size to realize the correct total dimension of  $\Delta\left(p_k\right)$. This interconnection is equivalent to a LPV-SS representation with rational dependency on $p_k$:
\begin{equation}
\begin{bmatrix}
    x_{k+1}\\
    y_k
\end{bmatrix} = 
\begin{bmatrix}
\mathcal{A}(p_k) & \mathcal{B}(p_k)\\
\mathcal{C}(p_k) & \mathcal{D}(p_k)
\end{bmatrix}
\begin{bmatrix}
    x_k\\
    u_k
\end{bmatrix}
,
\end{equation} where $\mathcal{A}(p_k), \mathcal{B}(p_k), \mathcal{C}(p_k), \mathcal{D}(p_k)$ are matrix functions of $p_k$.

This is achieved by eliminating the latent variables $w_k$ and $z_k$ in \eqref{LPVLFR_eq} by substitution of
\begin{equation}
\label{LFR_latent_subs}
z_k = (I-D_{zw}\Delta(p_k))^{-1}(C_zx_k+D_{zu}u_k).
\end{equation}
Which requires $(I-D_{zw}\Delta(p_k))$ to be non-singular. This gives rise to the following definition:
\begin{defn} \label{WP_defn} The LPV-LFR model defined by the pair $\{M, \Delta(p_k)\}$ according to \eqref{LPVLFR_eq} is \emph{well-posed}, if $I - D_{zw}\Delta(p_k)$ is non-singular, i.e., $\det(I-D_{zw}\Delta(p_k)) \!\neq\! 0$ for all possible realizations of the scheduling variable, i.e., $p_k \in \mathbb{P}$.
\end{defn}
As this representation admits rational LPV-SS models, it is a generalization of commonly used affine LPV-SS models. Notably, affine-dependency models are represented by taking \eqref{LPVLFR_eq} with $D_{zw} \!=\! 0$, resulting in a model set which is always well-posed. Furthermore, given an LPV-SS model with affine dependency an equivalent LPV-LFR realization can easily be recovered using \emph{Singular Value Decomposition} (SVD) methods \citep{zhou1996}.

LPV-LFR models can be used to represent LPV systems by considering $p_k$ to be an exogeneous signal. Nonlinear systems such as $\Sigma_\mathrm{o}$ in \eqref{NLsys_eq} require defining $p_k = \psi(x_k, u_k)$ where  $\psi : \mathbb{R}^{n_\mathrm{x}} \times \mathbb{R}^{n_\mathrm{u}} \to \mathbb{R}^{n_\mathrm{p}}$ is a \emph{scheduling map}, corresponding to a LPV embedding of the original nonlinear system.

\subsection{Identification structure}
We aim to identify the data-generating NL system $\Sigma_\mathrm{o}$ based on a dataset $\mathcal{D}_N$,  in the form of parameterized surrogate model. Specifically, we consider \emph{self-scheduled} LPV-LFR models, requiring the joint-estimation of LPV model and scheduling map, similar to \cite{verhoek2022}, which treats the LPV-SS setting.

We parameterize the surrogate model for the system \eqref{NLsys_eq} in an LPV-LFR form according to:
\begin{equation}
\label{LPVLFR_plant}
G_\theta: \left\{
\begin{aligned}
\begin{bmatrix}
\hat{x}_{k+1} \\
\hat{z}_k\\
\hat{y}_k
\end{bmatrix} &=
\begin{bmatrix}
A_x & B_w & B_u\\
C_z & D_{zw} & D_{zu}\\
C_y & D_{yw} & D_{yu}
\end{bmatrix}
\begin{bmatrix}
\hat{x}_k\\
\hat{w}_k\\
u_k
\end{bmatrix},\\
\hat{w}_k &= \Delta(\hat{p}_k)\hat{z}_k,
\end{aligned}
\right.
\end{equation}

where the matrices $A_x,\dots, D_{yu}$ form the parameter vector $\theta_G$. Next, we consider the scheduling signal $\hat{p}_k$ to be defined by a scheduling map $\hat{\psi}_\theta(x_k, u_k, d_k)$, where $d_k \in \mathbb{R}^{n_\mathrm{d}}$ is a known exogenous signal. The inclusion of $d_k$ allows for full or partial external scheduling. To be able to capture the set of data-generating systems described in \eqref{NLsys_eq}, the scheduling map is required to be a general function approximator. To achieve this, we propose the use of a \emph{Feedforward Neural Network} (FNN) parameterization of the scheduling map. Typically the FNN is extended with a linear bypass layer, resulting in a one-block ResNet \citep{he2016}. This bypass aids in the learning of linear dependencies on $x_k, u_k$ and $d_k$ if present. The weights and biases of the ResNet form the parameter vector $\theta_\psi$. For mildly nonlinear systems, the authors suggest the use of a ResNet with two hidden layers, limited amount of neurons per layer (6-12) and hyperbolic tangent (tanh) activation functions.
\vspace{-5pt}
\section{Ensuring well-posedness} \label{sec: WP}
\subsection{Well-posedness conditions}
As we have shown in \eqref{LFR_latent_subs}, for an LPV-LFR to be well-posed, we require $I - D_{zw}\Delta(p_k)$ to be non-singular for all possible realizations of the scheduling variables, i.e., $p_k \in \mathbb{P}$. We propose a novel method for guaranteeing well-posedness of LPV-LFR models, we begin by establishing key assumptions and conditions that play a role in this result.
\begin{assum}
\label{assumption_2}
The matrix function $\Delta(p_k)$ has a diagonal structure as defined in \eqref{Delta_structure}.
\end{assum}

\begin{assum} \label{assumption_3} The scheduling set\footnote{One could also assume constraints directly on the scheduling map, however by constraining the scheduling set the results also hold for externally scheduled LPV models.} $\mathbb{P}$ is contained in the closed $n_\mathrm{p}$-dimensional $\ell_\infty$ unit ball, i.e.,
\begin{equation}
    \mathbb{P} \subseteq \{p:\mathbb{Z}\to \mathbb{R}^{n_\mathrm{p}} \mid \|p_k||_\infty \leq 1\ \forall k \in \mathbb{Z}\}.
\end{equation}
\end{assum}
\label{remark_scaleshift}
It is important to note that Assumption \ref{assumption_3}
is not restrictive. Given an arbitrary, well-posed, LFR $\{G, \Delta(p) \}$, under Assumption \ref{assumption_2}
it is trivial to retrieve an equivalent LFR defined by the pair $\{ \bar{M}, \bar{\Delta}(p) \}$ which satisfies Assumption \ref{assumption_3}
by scaling and shifting the scheduling variables (see e.g., \cite{schererLMIbook}). For models with NN scheduling map this can be enforced through the incorporation of an activation function on the output layer, such as the $\tanh$ function. 
\begin{defn} The \emph{spectral radius} $\rho$ of a matrix $A \in \mathbb{R}^{n \times n}$ with eigenvalues $\lambda_1, \dots \lambda_n$ is defined as:
\begin{equation}
        \rho(A) = \max\{|\lambda_1|, \dots, |\lambda_n|\}.
\end{equation}
\end{defn}
\begin{cond}
\label{condition_5} 
The spectral radius $\rho$ of $D_{zw}$ is strictly smaller than 1, i.e.,    
\begin{equation}
    \rho(D_{zw}) < 1.
\end{equation}
\end{cond}

These lead to the following theorem.

\begin{thm} \label{theorem: WP theorem}
If Assumptions  \ref{assumption_2}, \ref{assumption_3}, and Condition \ref{condition_5}
hold, the LPV-LFR model \eqref{LPVLFR_eq} is well-posed.
\end{thm}

\begin{pf} From Assumptions \ref{assumption_2} and \ref{assumption_3} it follows that \\$\rho(\Delta(p_k)) \leq 1, \ \forall p_k \in \mathbb{P}$. Additionally, under Assumption \ref{assumption_2}, it holds that  $ \rho(D_{zw}\Delta(p_k))\leq\rho(D_{zw})\rho(\Delta(p_k))$. Thus, it follows that $I -D_{zw}\Delta(p_k) \succ 0$ for all combinations of $ p_k \in \mathbb{P}$. This implies $\det(I-D_{zw}\Delta(p_k)) > 0, \ \forall p_k \in \mathbb{P}$, satisfying the well-posedness condition. 
\end{pf} 
\subsection{Well-posedness by direct parameterization}
To ensure Condition \ref{condition_5} is satisfied, we define
\begin{equation}
\label{Dzw_definition}
D_{zw} = e^{-N}, \quad N \succ 0.
\end{equation} As the matrix exponential maps  eigenvalues $\lambda$ to $e^\lambda$, the originally negative eigenvalues of $-N$ are mapped to eigenvalues of $D_{zw}$ that are strictly smaller in magnitude than one. Based on \eqref{Dzw_definition} we propose direct parameterization satisfying $N \succ 0$, enabling unconstrained optimization.

To guarantee Condition \ref{condition_5}
through $D_{zw}$ defined according to \eqref{Dzw_definition} with $N \succ 0$, we use the parameterization
\begin{equation}
    \label{construction_N}
    N = \Psi(D_A^\top D_A + D_B-D_B^\top + \epsilon I),
\end{equation}

where the matrix $\Psi \!=\! \textrm{diag}(e^{D_d})$ such that $N$ is constructed with the free variables  $D_A \in \mathbb{R}^{n_\mathrm{w}\times n_\mathrm{w}}, D_B \in \mathbb{R}^{n_\mathrm{w}\times n_\mathrm{w}} $ and $D_d \in \mathbb{R}^{n_\mathrm{w}}$ and a constant $0 < \epsilon \ll 1$. By combining a symmetric positive-semi-definite term $D_A^\top D_A$, a skew-symmetric term $D_B-D_B^\top$ and a strictly positive definite term $\epsilon I$, with positive scaling by $\Psi$, $N$ is guaranteed to be strictly positive-definite by construction. The positive scaling term $\Psi$ is shown in \cite{revay2020} to extend the feasible domain of $N$. Additionally, it aids in counteracting potential poor scaling of variables which can be introduced through the transformation \eqref{Dzw_definition}.

When considering the use of LPV-LFR models in identification, the parameterization of $N$ according to \eqref{construction_N} enables unconstrained optimization without risk of obtaining ill-posed models. By substituting the direct free parameterization of $D_{zw}$ with the transformed free variables in $D_{A}, D_B$ and $D_d$ for constructing $D_{zw}$, the number of parameters is increased from $n_\mathrm{w}^2$ to $2n_\mathrm{w}^2 +n_\mathrm{w}$. However, as noted in \cite{winston2020}, by populating only the lower triangular elements of $D_A$ and the strictly upper triangular elements of $D_B$, this overparameterization is greatly reduced, now requiring only $n_\mathrm{w}^2 + n_\mathrm{w}$ elements, $n_\mathrm{w}$ variables more than populating $D_{zw}$ directly.

\section{Identification Method} \label{sec: ident_method}
\subsection{Parameter Estimation}
We aim to estimate a parametric surrogate model using LFR parameterization \eqref{LPVLFR_plant}, given a dataset $\mathcal{D}_N$ by minimizing the criterion:
\begin{equation}
\label{datafit_cost}
V_N^\mathrm{fit}(\theta,\mathcal{D}_N) = \frac{1}{N}\sum^{N-1}_{k=0}\|\hat{e}_{k|\theta}\|_2^2,
\end{equation}
where $\hat{e}_{k|\theta} = y_k - \hat{y}_{k|\theta}$ represents the prediction error given the parameters $\theta$. By choosing the model structure in \eqref{LPVLFR_plant}, with \emph{Output-Error} (OE) noise model, the prediction error equals the simulation error.

The parameter vector $\theta$ is comprised of the combined parameter vectors of the plant and scheduling map, $\theta = [\theta_G^\top  \ \theta_\psi^\top]^\top $. The minimization of \eqref{datafit_cost}  matches the classical PEM criterion which aims to minimize the \emph{Mean Squared Error} (MSE) of the predictions. This leads to the following optimization problem:
\begin{align}
    \underset{\theta}{\mathrm{argmin}} \quad &  V_N^\mathrm{fit}(\theta,\mathcal{D}_N) + \rho\|\theta\|_2^2,\notag\\
    \mathrm{s.t.} \quad &  
      x_0 \!=\! \theta_{x_0}, \notag \\
     &\hat{p}_k \!=\! \hat{\psi}(\hat{x}_k,u_k,d_k,\theta),\\
     & \hat{x}_{k+1} \!=\!  \mathcal{A}(\Delta(\hat{p}_k),\theta)\hat{x}_k \!+\! \mathcal{B}(\Delta(\hat{p}_k),\theta)u_k,\notag\\    
     & \hat{y}_k  \!= \mathcal{C}(\Delta(\hat{p}_k),\theta)\hat{x}_k \!+\! \mathcal{D}(\Delta(\hat{p}_k),\theta)u_k,\notag\\
     & \hat{e}_k  \!=\! y_k - \hat{y}_k, \notag
\end{align}
where $\theta_{x_0}$ represents the initial state as additional decision variable and where  $\rho\|\theta\|_2^2$ is added as an $\ell_2$-regularization term. 
To solve this optimization problem, we leverage the approach described in \cite{bemporad2025a} based on JAX (\cite{jax2018github}). In this gradient-based method, local optimization is performed in two steps. First, a fixed number of iterations is taken using the Adam algorithm
, followed by a number of iterations using the \emph{bound-constrained limited-memory Broyden-Fletcher-Goldfarb-Shanno} (L-BFGS-B) algorithm (\cite{liu1989}). The use of Adam, through the incorporation of \emph{momentum}, aids in escaping local minima at the early iterations of the optimization process, providing a better point of initialization for the L-BFGS-B optimizer, which subsequently converges accurately to the final minimum of the optimization process. This two-step approach has been shown to be computationally efficient for the identification of NL-SS systems (\cite{bemporad2025a}) as well as self-scheduled LPV-SS models (\cite{bemporad2025b}), particularly when compared with methods using Adam exclusively.

\subsection{Parameter Initialization}
When starting the optimization procedure, initial values for all elements in the parameter vector are required. This is done in a pseudo-random manner. The elements of the plant are initialized as follows:
\begin{equation}
\label{plant_random_init}
\begin{array}{ll}
    A_x = 0.5I, \quad  & B_w = \mathbf{0}, \\
    B_u \sim  \mathcal{U}(-0.1, 0.1), \quad &  C_z \sim \mathcal{U}(-0.1, 0.1),\\
    D_{zu} \sim \mathcal{U}(-0.1, 0.1), \quad  & C_y \sim \mathcal{U}(-0.1, 0.1), \\
    D_{yw} = \mathbf{0}, \quad &D_{yu} = \mathbf{0}.
\end{array}
\end{equation}
If affine LPV-LFR models are used, $D_{zw}$ is left out of the parameter vector, fixed at a zero matrix. For well-posed rational LPV-LFR models,
we construct the $D_{zw}$ matrix by drawing lower triangular elements from a uniform distribution $\mathcal{U}(0 , 0.1)$ for the matrix $D_A$, the strictly upper triangular elements of $D_B$ follow the normal distribution $\mathcal{N}(1,1)$, finally $D_d = \mathbf{0}$ is chosen to initialize the remainder of the parameter vector. Lastly, to initialize the parameters of the scheduling map, represented by a ResNet, uniform Xavier initialization (\cite{glorot2010}) is used. Bias terms are initialized as zero.

\section{Simulation Examples} \label{sec: sim_examples}

To evaluate the performance of the identified models, we consider two simulation examples. First, a \emph{Nonlinear Mass-Spring-Damper} (NL-MSD) system is considered to evaluate the trade-off of rational-scheduling dependency. Secondly, we consider a \emph{Control Moment Gyroscope} (CMG), previously used to evaluate performance of joint-estimation methods, enabling comparison of the proposed methodology to existing methods. As a performance measure, the \emph{Best Fit Rate} (BFR) is considered, defined as:
\begin{equation}
\label{BFR}
   \mathrm{BFR} = \max\left\{1- \frac{\sum^{N-1}_{k=0}\|y_k - \hat{y}_k\|_2}{\sum^{N-1}_{k=0}\|y_k - \bar{y}\|_2},\  0 \right\}\cdot 100\%,
\end{equation}
where $\hat{y}_k$ is the estimated output and $\bar{y}$ is the sample mean of $y$.

The LPV-LFR models of both examples are trained according to the same optimization strategy proposed in Section \ref{sec: ident_method}. Leveraging the efficiency of the JAX-based method, all models are trained for 100 random initializations on an Intel 14700K CPU, and the model with the highest BFR on validation data is selected.

\subsection{Nonlinear Mass-Spring-Damper}
To illustrate the trade-off between affine and rational-dependency LPV-LFR models  we show a numerical example based on a DT NL-MSD system described by:
\begin{align}
\label{DT_NL_MSD_eq}
x_{1,k+1} \!&=\! x_{1,k} \!+\! T_sx_{2,k},\notag\\
x_{2,k+1} \!&=\! x_{2,k} \!+\! \tfrac{T_s}{m}\left(u_k \!-\! \tfrac{3}{5}x_{1,k}u_k \!-\! k_1 x_{1,k} \!-\! k_2x_{1,k}^3 \!-\! d_1 x_{2,k} \right),\notag\\
y_k \!&=\! x_{1,k},
\end{align}
with $T_s \! =\! 0.1$, $m \!= \! 1$, $k_1 \!=\! 0.1$, $k_2 \!=\! 1$, $d_1 \! = \! 1$ and where $x_1$ represents the position of the system and $x_2$ its velocity.
This plant features nonlinearities dependent on $x_1$ and $x_1^2$, requiring two scheduling variables for affine LPV representations or one scheduling variable for rational dependency models.

\emph{The Data:   }
Using the DT system, we generate three datasets, one training, validation and test set. Each dataset is generated by exciting the system \eqref{DT_NL_MSD_eq} with a different realization of a random phase multisine signal. The training and validation input signals are $N = 6 \cdot10^3$ samples long, the test set is longer, consisting of $N 
\!=\! 3 \cdot 10^4$ samples. Each input has frequency components from $\tfrac{1}{6}$ Hz to $5$ Hz. The input signals have amplitudes with standard deviation $\sigma_u = 4$ N. The output of \eqref{DT_NL_MSD_eq} has additive white noise $\mathcal{N}(0,\sigma_e^2)$ with $\sigma_e^2 = 0.063$, resulting in a SNR of approximately $20$ dB.

\emph{Model Settings:   }
To compare the performance of rational and affine-dependency based LPV-LFR models, we train such models using the well-posed parameterization with ($n_\mathrm{x} \!=\! 2, n_\mathrm{p} \!=\! 1, \eta \!=\! 3$), matching the required order necessary for rational models. Subsequently we train affine models with ($n_\mathrm{x} \!=\! 2, n_\mathrm{p} \!=\! 2, \eta \!=\! 1$) matching the theoretically required dimensions for affine models. On this example, linear scheduling maps with saturation: 
\begin{equation}
\label{eq: MSD_schedmap}
    \hat{\psi}(x,u) = \tanh(W_\mathrm{x}x + W_\mathrm{u}u + b),
\end{equation}
 are considered. Here $W_\mathrm{x}, W_\mathrm{u},b$ are free variables.\ Each model is trained for 1000 Adam epochs followed by 4000 epochs using L-BFGS.
 
\emph{Results:   }
The results achieved on the NL-MSD benchmark are shown in Table \ref{tb:NLMSD_results} and illustrated in  Figure \ref{fig:MSD_results}. They show that LPV-LFR models with affine and rational scheduling dependency are capable of achieving accurate results on test data-set. However, given the simple scheduling map \eqref{eq: MSD_schedmap}, the affine models require a higher scheduling dimension $n_\mathrm{p} = 2$ to achieve the same accuracy as the LPV-LFR model with rational scheduling dependency achieves with a scheduling dimension of one. This illustrates the trade-off between complexity in terms of scheduling dimension and complexity in terms of scheduling dependency which is made possible through the well-posed parameterization introduced in this work.

\begin{figure}[!b]
\begin{center}
\includegraphics[width=8.4cm]{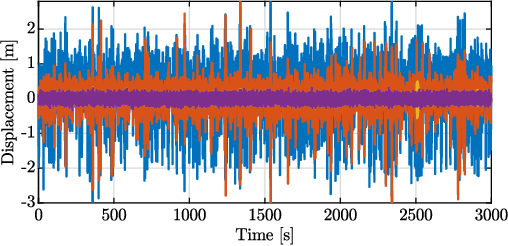}   
\caption{True output (\textcolor{matlabblue}{\raisebox{0.5mm}{\rule{0.5cm}{0.6mm}}}), affine LPV-LFR ($n_\mathrm{p} \!=\! 1$) sim.  error (\textcolor{matlabred}{\raisebox{0.5mm}{\rule{0.5cm}{0.6mm}}}), affine LPV-LFR ($n_\mathrm{p} \!=\! 2$) sim. error (\textcolor{matlabyellow}{\raisebox{0.5mm}{\rule{0.5cm}{0.6mm}}}), and rational LPV-LFR sim. error (\textcolor{matlabpurple}{\raisebox{0.5mm}{\rule{0.5cm}{0.6mm}}}) on the test data of the NL-MSD.} 
\label{fig:MSD_results}
\end{center}
\end{figure}

\begin{table}[!t]
\begin{center}
\caption{Results on the NL-MSD example using LPV-LFR models.}\label{tb:NLMSD_results}
\begin{tabular*}{\columnwidth}{@{\extracolsep{\fill}} l c c c r}
Model  & BFR train& BFR val & BFR test & time \\\hline

True System & 89.47 & 90.19 & 90.15 & - \\
Affine $\!(n_\mathrm{p} \!=\! 1)\!$& 53.56 & 41.78 & 44.90 &   37 s\\
Affine $\!(n_\mathrm{p} \!=\! 2)\!$& 89.56 & 90.20 & 89.87 &   87 s\\
Rational $\!(n_\mathrm{p} \!=\! 1)\!$ & 89.57 & 90.20 & \bf{90.16}  & 112 s \\ \hline
\end{tabular*}
\end{center}
\end{table}

\subsection{Control Moment Gyroscope}
\begin{figure}[!t]
\begin{center}
\includegraphics[width=2.4cm]{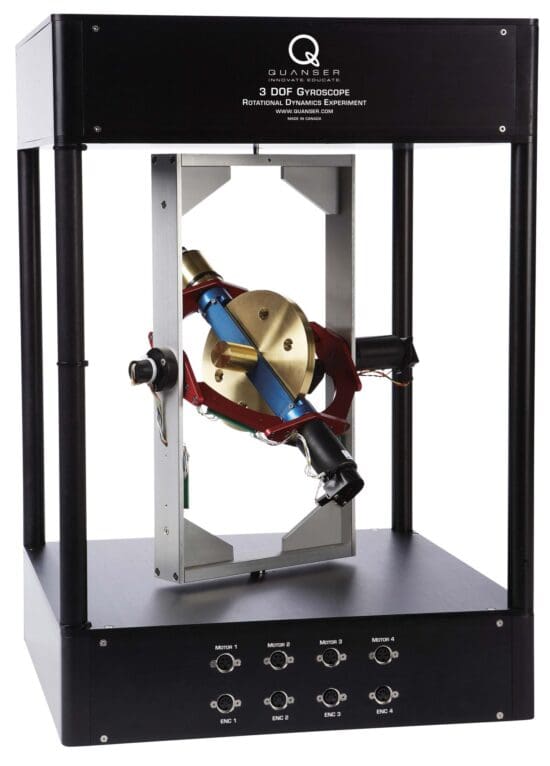}    
\caption{Picture of the CMG considered in example 2.} 
\label{fig: gyroscope}
\end{center}
\end{figure}
As second simulation example, we consider the 3-DOF CMG and compare achieved performance with the methods presented in \cite{verhoek2022} and \cite{bemporad2025b}. The 3-DOF CMG, as shown in Figure \ref{fig: gyroscope}, consists of a disk mounted inside a blue gimbal, which in turn is mounted inside a red gimbal which is mounted inside of silver gimbal. The disk, and gimbals are each equipped with a DC motor for actuation and an encoder to measure the angular positions ($q_1, q_2, q_3$ and $q_4$ for the disk, blue, red, and silver gimbal respectively). During the simulation, the red gimbal is locked at $q_3 = 0$, the blue gimbal $q_2$ is actuated with a motor current $i_2$, and the angular velocity of the silver gimbal $\dot{q}_4$ is measured and considered as output. The velocity of the disk $\dot{q}_1$ is controlled independently, tracking a random multi-level-reference signal with amplitudes between 30 and 50 rad/s and dwell time of 4-8 seconds. The resulting dynamics can, for small angles of $q_2$, be interpreted as a  nonlinear SISO model between the input $i_2$ and output $\dot{q}_4$ with dynamics dependent on the external input $\dot{q}_1$.

\emph{The Data:   }
As \cite{verhoek2022} describes, the input of the first gimbal is excited with a sinusoidal carrier wave with white noise added: $u_k\!\!=\! \!\tfrac{1}{2}\sin(\omega T_sk)\! +\! v_k$ where $v_k \!\sim\! \mathcal{N}(0,\sigma_v^2)$ and $\sigma_v \!=\! \frac{1}{3}$. For each dataset, the frequency of the sinusoidal carrier wave is drawn from a uniform distribution, $\omega \!\sim\! \mathcal{U}(1,2)$. Furthermore, an OE noise structure is considered with error signal $e_k \!\sim\! \mathcal{N}(0,\sigma_e^2)$ with variance $\sigma_e^2 \!=\! 2.2 \cdot 10^{-5}$ corresponding to a SNR of $35$ dB. The velocity of the disk follows the previously specified reference. The data is generated using CT simulation with RK4 integration and is subsequently down-sampled to $100$ Hz. A training set consisting of $10^4$ samples is created as well as a validation set of $3 \cdot 10^4$ samples.

\emph{Model Settings:   }
To compare the performance of the proposed LPV-LFR-based approach with the methods described in \cite{verhoek2022} and \cite{bemporad2025b}, we opt to train models with the same state and scheduling dimension ($n_\mathrm{x} \! = \! 5, n_\mathrm{p} \!=\! 3$).  Rational-dependency and affine-dependency-based models are trained for comparison, both with $\eta = 1$. As scheduling map, 2-hidden layer ResNets with 6 neurons per hidden layer and tanh activation functions are used. Each model is trained for 2000 epochs using Adam, followed by 4000 epochs using L-BFGS.

\emph{Results:   }
The simulation performance of the LTI, affine LPV-LFR and rational LPV-LFR model are shown in Figure \ref{fig:gyro_results} and compared with the results in \cite{bemporad2025b} in Table \ref{tb:CMG_results}.
The results show that for affine models, results similar to existing joint-estimation methods can be achieved, for rational models however, performance is slightly worse. This reduction in performance is attributed to the increased complexity of the optimization problem, resulting form the additional complexity of the rational scheduling dependency. 
\begin{figure}[!b]
\begin{center}
\includegraphics[width=8.4cm]{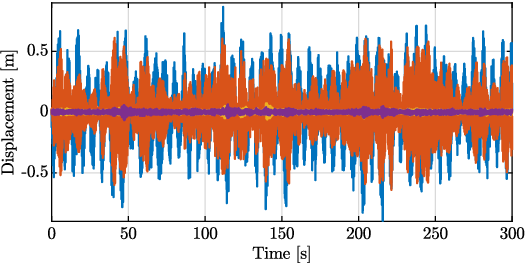}    
\caption{True output (\textcolor{matlabblue}{\raisebox{0.5mm}{\rule{0.5cm}{0.6mm}}}), LTI sim. error (\textcolor{matlabred}{\raisebox{0.5mm}{\rule{0.5cm}{0.6mm}}}), affine LPV-LFR sim. error (\textcolor{matlabyellow}{\raisebox{0.5mm}{\rule{0.5cm}{0.6mm}}}), and rational LPV-LFR sim. error (\textcolor{matlabpurple}{\raisebox{0.5mm}{\rule{0.5cm}{0.6mm}}}) on the validation data of the CMG.} 
\label{fig:gyro_results}
\end{center}
\end{figure}

\begin{table}[!t]
\begin{center}
\caption{Comparison of LPV identification results on the CMG dataset. Previous results are taken from \cite{bemporad2025b} and achieved on different hardware, computation times may differ. The LPV-SUBNET method is described in \cite{verhoek2022}.}\label{tb:CMG_results}
\begin{tabular*}{\columnwidth}{@{\extracolsep{\fill}} l c c c r}
Model  & BFR train & BFR val  & time \\\hline
True System & 98.16 & 98.19  & - \\
Affine LPV-LFR & 97.46 & 96.41  &   172 s\\
Rat. LPV-LFR & 97.12 & 96.02  & 279 s \\
LPV-SS & 97.61 & 96.50  & 47.54 s \\
LPV-SUBNET & 97.28 & 96.40  & $\approx$ 10 h \\ \hline
\end{tabular*}
\end{center}
\end{table}

\section{Conclusion} \label{sec: conclusion}
This work has proposed the use of the linear fractional representation for LPV models in the joint-estimation of LPV models and scheduling maps. A novel parameterization which, by construction, guarantees well-posedness of LPV-LFR models with rational scheduling dependency, has been introduced. The well-posed parameterization has been shown to be capable of accurately modelling nonlinear systems a simulation example, illustrating how the rational scheduling dependency can reduce the required scheduling order.  Finally, estimation of LPV-LFR models with affine and rational-dependency has been compared to state-of-the-art joint-estimation methods for LPV models on a benchmark where LPV-LFR models have been shown to achieve results close to existing methods.

\bibliography{ifacconf} 

@book{mohammadpour2012,
  title = {Control of {{Linear Parameter Varying Systems}} with {{Applications}}},
  editor = {Mohammadpour, Javad and Scherer, Carsten W.},
  year = {2012},
  publisher = {Springer US},
  address = {Boston, MA},
  copyright = {https://www.springernature.com/gp/researchers/text-and-data-mining},
  isbn = {978-1-4614-1832-0 978-1-4614-1833-7},
  langid = {english}
}

@inproceedings{hjartarson2015,
  title = {{{LPVTools}}: {{A Toolbox}} for {{Modeling}}, {{Analysis}}, and {{Synthesis}} of {{Parameter Varying Control Systems}}},
  shorttitle = {{{LPVTools}}},
  author = {Hjartarson, Arnar and Seiler, Peter and Packard, Andrew},
  year = {2015},
  booktitle = {Proc. of the 1st IFAC Workshop on Linear Parameter Varying Systems},
  address = {Grenoble, France},
  pages = {139--145},
  issn = {24058963},
  copyright = {https://www.elsevier.com/tdm/userlicense/1.0/},
  langid = {english}
}

@inproceedings{boef2021,
  author = {Pascal {den Boef} and Pepijn B. Cox and Roland Tóth},
  title = {{LPVcore}: {MATLAB} toolbox for {LPV} modelling, identification and control},
  booktitle = {Proc. of the 19th IFAC Symposium on System Identification},
  pages = {385--390},
  year = {2021},
  issn = {2405-8963}}

@article{vanwingerden2009,
  title = {Subspace Identification of {{Bilinear}} and {{LPV}} Systems for Open- and Closed-Loop Data},
  author = {Van Wingerden, Jan-Willem and Verhaegen, Michel},
  year = {2009},
  journal = {Automatica},
  volume = {45},
  number = {2},
  pages = {372--381},
  issn = {00051098},
  copyright = {https://www.elsevier.com/tdm/userlicense/1.0/},
  langid = {english}
}

@article{cox2021,
  title = {Linear Parameter-Varying Subspace Identification: {{A}} Unified Framework},
  shorttitle = {Linear Parameter-Varying Subspace Identification},
  author = {Cox, Pepijn Bastiaan and T{\'o}th, Roland},
  year = {2021},
  journal = {Automatica},
  volume = {123},
  pages = {109296},
  issn = {00051098},
  langid = {english}
}

@inproceedings{verhoek2022,
  title = {Deep-{{Learning-Based Identification}} of {{LPV Models}} for {{Nonlinear Systems}}},
  booktitle = {Proc. of the 61st {{IEEE}} {{Conference}} on {{Decision}} and {{Control}} },
  author = {Verhoek, Chris and Beintema, Gerben I. and Haesaert, Sofie and Schoukens, Maarten and Toth, Roland},
  year = {2022},
  pages = {3274--3280},
  address = {Cancun, Mexico},
  langid = {english}
}

@article{bemporad2025a,
  title={An {L-BFGS-B} Approach for Linear and Nonlinear System Identification Under $\ell_1$ and Group-Lasso Regularization},
  author={Bemporad, Alberto},
  journal={IEEE Transactions on Automatic Control},
  year={2025},
  publisher={IEEE}
}

@article{zhou1996,
  title={Robust and Optimal Control},
  author={Zhou, K and Doyle, JC and Glover, K},
  journal={Control Engineering Practice},
  volume={4},
  number={8},
  pages={1189--1190},
  year={1996},
  publisher={Elsevier Science Publishing Company, Inc.}
}

@article{bemporad2025b,
  title = {Efficient Identification of Linear, Parameter-Varying, and Nonlinear Systems with Noise Models},
  author = {Bemporad, Alberto and T{\'o}th, Roland},
  year = {2025},
  journal = {arXiv preprint arXiv:2504.11982},
  eprint = {2504.11982},
  primaryclass = {math},
  publisher = {arXiv},
  archiveprefix = {arXiv},
  langid = {english}
}

@article{schererLMIbook,
  title={Linear Matrix Inequalities in Control},
  author={Scherer, Carsten and Weiland, Siep},
  year={2015},
  journal={Lecture Notes}
}

@article{revay2020,
  title = {Lipschitz {{Bounded Equilibrium Networks}}},
  author = {Revay, Max and Wang, Ruigang and Manchester, Ian R.},
  year = {2020},
  journal = {arXiv preprint arXiv:2010.01732},
  eprint = {2010.01732},
  primaryclass = {cs},
  publisher = {arXiv},
  archiveprefix = {arXiv},
  langid = {english}
}

@article{winston2020,
  title = {Monotone Operator Equilibrium Networks},
  author = {Winston, Ezra and Kolter, J Zico},
  year = {2020},
  journal = {Advances in neural information processing systems},
  volume = {33},
  pages = {10718--10728},
  langid = {english}
}

@inproceedings{glorot2010,
  title={Understanding the difficulty of training deep feedforward neural networks},
  author={Glorot, Xavier and Bengio, Yoshua},
  booktitle={Proc. of the 13th {International} {Conference} on {Artificial} {Intelligence} and {Statistics}},
  pages={249--256},
  address = {Sardinia, Italy},
  year={2010},
  organization={JMLR Workshop and Conference Proceedings}
}

@inproceedings{hoffmann2014,
  title = {Complexity of {{Implementation}} and {{Synthesis}} in {{Linear Parameter-Varying Control}}},
  author = {Hoffmann, Christian and Werner, Herbert},
  year = {2014},
  booktitle = {Proc. 19th {IFAC} World Congr.},
  address = {Cape Town, South Africa},
  issn = {14746670},
  copyright = {https://www.elsevier.com/tdm/userlicense/1.0/},
  langid = {english}
}

@inproceedings{nemani1995,
  title = {Identification of Linear Parametrically Varying Systems},
  booktitle = {Proc. of the 34th {{IEEE Conference}} on {{Decision}} and {{Control}}},
  author = {Nemani, M. and Ravikanth, R. and Bamieh, B.A.},
  year = {1995},
  pages = {2990--2995},
  address = {New Orleans, LA, USA},
  isbn = {978-0-7803-2685-9},
  langid = {english}
}

@article{lee1999,
  title = {Identification of {{Linear Parameter-Varying Systems Using Nonlinear Programming}}},
  author = {Lee, Lawton H. and Poolla, Kameshwar},
  year = {1999},
  journal = {Journal of Dynamic Systems, Measurement, and Control},
  volume = {121},
  number = {1},
  pages = {71--78},
  langid = {english}
}

@inproceedings{cheng2015,
  title = {Identification of {{LPV}} Systems with {{LFT}} Parametric Dependence via Convex Optimization},
  booktitle = {Proc. of the 54th {{IEEE Conference}} on {{Decision}} and {{Control}}},
  author = {Cheng, Yongfang and Sznaier, Mario},
  year = {2015},
  pages = {1459--1464},
  address = {Osaka, Japan},
  isbn = {978-1-4799-7886-1},
  langid = {english}
}

@inproceedings{mejari2019,
  title = {Kernelized {{Identification}} of {{Linear Parameter-Varying Models}} with {{Linear Fractional Representation}}},
  booktitle = {Proc. of the 18th {{European Control Conference}} },
  author = {Mejari, Manas and Piga, Dario and Toth, Roland and Bemporad, Alberto},
  year = {2019},
  pages = {337--342},
  address = {Naples, Italy},
  langid = {english}
}

@inproceedings{he2016,
  title = {Deep {{Residual Learning}} for {{Image Recognition}}},
  booktitle = {Proc. of the {{IEEE Conference}} on {{Computer Vision}} and {{Pattern Recognition}}},
  author = {He, Kaiming and Zhang, Xiangyu and Ren, Shaoqing and Sun, Jian},
  year = {2016},
  pages = {770--778},
  address = {Las Vegas, USA},
  langid = {english}
}

@misc{jax2018github,
  author = {James Bradbury and Roy Frostig and Peter Hawkins and Matthew James Johnson and Chris Leary and Dougal Maclaurin and George Necula and Adam Paszke and Jake VanderPlas and Skye Wanderman-Milne and Qiao Zhang},
  title = {{JAX}: composable transformations of Python+NumPy programs},
  year = {2018},
  howpublished = {\url{https://github.com/google/jax}},
}

@article{liu1989,
  title = {On the Limited Memory {{BFGS}} Method for Large Scale Optimization},
  author = {Liu, Dong C. and Nocedal, Jorge},
  year = {1989},
  journal = {Mathematical Programming},
  volume = {45},
  number = {1-3},
  pages = {503--528},
  issn = {0025-5610, 1436-4646},
  copyright = {http://www.springer.com/tdm},
  langid = {english}
}

@book{toth2010,
  title = {Modeling and {{Identification}} of {{Linear Parameter-Varying Systems}}},
  author = {T{\'o}th, Roland},
  year = 2010,
  series = {Lecture {{Notes}} in {{Control}} and {{Information Sciences}}},
  volume = {403},
  publisher = {Springer},
  address = {Berlin, Heidelberg},
  copyright = {http://www.springer.com/tdm},
  isbn = {978-3-642-13811-9 978-3-642-13812-6},
  langid = {english}
}

@inproceedings{verhoek2023,
  title = {Learning {{Stable}} and {{Robust Linear Parameter-Varying State-Space Models}}},
  booktitle = {Proc. of the 62nd {{IEEE Conference}} on {{Decision}} and {{Control}}},
  author = {Verhoek, Chris and Wang, Ruigang and T{\'o}th, Roland},
  year = 2023,
  month = dec,
  eprint = {2304.01828},
  primaryclass = {eess},
  pages = {1348--1353},
  langid = {english},
  address = {Singapore},
}

@article{bianchi2010,
  title={Robust identification/invalidation in an {LPV} framework},
  author={Bianchi, Fernando D and S{\'a}nchez-Pe{\~n}a, Ricardo S},
  journal={International Journal of Robust and Nonlinear Control},
  volume={20},
  number={3},
  pages={301--312},
  year={2010},
  publisher={Wiley}
}

@article{scherer2001,
  title = {{{LPV}} Control and Full Block Multipliers},
  author = {Scherer, C W},
  year = 2001,
  journal = {Automatica},
  langid = {english},
  volume = 37,
  pages = {361--375}
}
\appendix

\end{document}